\documentclass[12pt]{article}
\usepackage{amssymb,amsmath,color,natbib,graphicx,amsthm,
  setspace,sectsty,anysize,times,dsfont,enumerate}

\usepackage{lscape,arydshln,relsize,rotating}
\usepackage[small]{caption}

\newtheorem{prop}{\sc Proposition}[section]

\marginsize{1.1in}{.9in}{.3in}{1in}

\newcommand{\bs}[1]{\boldsymbol{#1}}

\newcommand{\mr}[1]{\mathrm{#1}}
\newcommand{\bm}[1]{\mathbf{#1}}
\newcommand{\ds}[1]{\mathds{#1}}
\newcommand{\indep}{\perp\!\!\!\perp}
\DeclareMathOperator*{\argmin}{argmin}

\sectionfont{\noindent\normalfont\normalsize\bf}
\subsectionfont{\noindent\normalfont\sc}
\subsubsectionfont{\noindent\normalfont\it}

\begin{document}

\setstretch{1}

\pagestyle{empty}

\noindent {\large \bf Rejoinder: Efficiency and Structure in MNIR}

\vskip .2cm

\noindent 
Matt Taddy, 
The  University of Chicago Booth School of Business
\vskip .1cm
\noindent \hrulefill

\vskip .75cm
\setstretch{1.4}

I thank Prof.~Blei and Grimmer for their comments; it is great to have one's
work discussed by researchers who are both excellent statisticians and experts
in their respective fields.

The discussion can be  summarized under two themes. Prof.~Blei  is interested
in extending MNIR  to modeling additional, often latent, structure in text.
Prof.~Grimmer is concerned with causation and interpretability.  Both will be
answered in context of my original motivation for MNIR:  the estimation
efficiency derived from  assumptions on $\bm{x}|y$.  We'll begin with
estimator properties in a simple illustration, then turn to discussion of
latent factors and causal inference.

\section{Efficiency}

A related question of efficiency  has been studied by \citet{Efro1975} and
\citet{NgJord2002} in comparisons between logistic regression and `generative'
discriminant analysis.  Efron's generative classifier applies Bayes rule to
inverse multivariate normals $\bm{x}|y \sim \mr{N}(\mu_y, \Sigma)$, where
$\mu_y = \ds{E}[\bm{x}|y]$ varies with $y\in \{0,1\}$ but the covariance
matrix is shared across populations.  Given true  normal covariate
distributions separated by root Mahalanobis distances of 3 to 4, he finds
predictions from this routine to be 1.5 to 3 times more efficient than
logistic regression.  This efficiency gain is smaller than that found by  Ng and
Jordan for a Naive Bayes algorithm (each covariate is fit as
independent of the others given $y$), with their results loosely interpreted to
imply $\log(n)$ times higher efficiency for the generative predictor. Although
Naive Bayes independence is not assumed for the data itself, requirements on the
amount of information about $y$ available in each covariate have the effect of
limiting conditional dependence.

Our model presents a third scenario:  covariate dependence is fully
specified via the negative correlation of a multinomial. Consider
binary response $y \in \{0,1\}$ and the joint word-sentiment  distribution
$\mr{p}(\bm{x},y) = \mr{MN}(\bm{x}\mid \bm{q}(y))\mr{p}(y)$ where $q_j(y) =
\mr{exp}[\alpha_j + \varphi_jy]/\sum_l \mr{exp}[\alpha_l + \varphi_ly]$ --
that is, the collapsed model in Equation 1 of the main paper.   Then
the expected information for $\bs{\varphi}$ is $\pi \bm{W}$, where $\pi =
\ds{E}[y]$ and $\bm{W} = \mr{diag}(\bm{q}_1) - \bm{q}_1\bm{q}_1'$ with
$\bm{q}_1 = \bm{q}({y}=1)$, and standard results \citep[e.g.,][chap.
5]{vdv1998} imply that in a fixed vocabulary the variance for maximum
likelihood estimator $\bs{\hat \varphi}$ scales with $M = \sum_i \sum_j
x_{ij}$, the  total number of  words.

{\samepage \begin{prop}\label{mlevar}   Assume the above joint model for $y$ and 
$\bm{x}$ with $\pi>0$, and write $\bs{\hat
\varphi}$  for the MLE fit of $\bs{\varphi}$ in our collapsed
MNIR model.  The estimation error converges
in distribution as \[ \sqrt{\pi M}(\bs{\hat \varphi}-\bs{\varphi}) \leadsto
\mr{N}\left(\bm{0}, \bm{W}^{-1}\right) \]  \end{prop} }

\noindent Thus variance decreases with the amount of speech rather
than with the number of speakers.  

Prediction requires an accompanying forward model.  If the collapsed model
holds true, Bayes rule implies a forward predictor and results of Proposition
\ref{mlevar} apply directly.  A more realistic scenario has the collapsed
model misspecified on an individual level.  Consider a model of individual
heterogeneity such that $\bm{x} \indep y \mid \bm{x}'\bs{\varphi}, \bm{u}$
where $\bs{\varphi}$ can be estimated consistently as in Proposition
\ref{mlevar} and $\bm{u}$ is a vector of unobserved random effects -- for
example, the model of Section 3.3 with $x_{ij} \sim
\mr{Po}\left(\mr{exp}[\mu_j + \varphi_jy_i + u_{ij}]\right)$ and $y_i \indep
u_{ij} \sim \mr{N}(0,1)$.  Write $z = \bs{\varphi}'\bm{f} =
\bs{\varphi}'(\bm{x}/m - \frac{1}{n}\sum_i \bm{x}_i/m_i)$ for   projection of
mean shifted frequencies $\bm{F} = [\bm{f}_1 \cdots \bm{f}_n]'$, and say
{MNIR-OLS} is the {two-stage} estimation of  $\bs{\hat \varphi}$ in collapsed
MNIR and $[\hat\alpha,\hat\beta]$ given  $\bm{\hat{z}} = \bm{F}\bs{\hat
\varphi}$ via least-squares (OLS).  Consider the simple forward approximation
$\ds{E}[y | \bm{f},\bm{u}] = \alpha + \beta z$ (e.g., if $y = \tilde \alpha +
\tilde \beta z + \bs{\gamma}' \bm{u} + \varepsilon$ and $u_j = a_j + b_j z +
\nu_j$ with $\nu_j\indep z$,  then $\beta = \tilde \beta +
\bs{\gamma}'\bm{b}$). Since $\ds{E}[y| \bm{f}] = \ds{E}[y| \bm{f}, \bm{u}]$
we have $\ds{E} \argmin_{\bs{\theta}} \sum_i (y_i - \alpha -
\bm{f}_i'\bs{\theta})^2   = \bs{\varphi}\beta$, such that OLS and MNIR-OLS
have the same  expectation and the  effect of $\bm{u}$ on $z$ is subsumed in
$\beta$.

The distinction of
MNIR-OLS is its estimation precision.

\begin{prop}\label{olsvar}     Consider data from  the joint word-sentiment
distribution of  Proposition \ref{mlevar} partitioned into documents
$\{\bm{x}_i,y_i\}_{i=1}^n$ where $0<\sum_i y_i<n$.   Assuming a finite 
upper-bound for each $|\hat\varphi_j|$, the MNIR-OLS predictor $\hat{y}(\bm{x})$
for a new document $\bm{x}$ has

\vspace{-.25cm}
\[  
\mr{var}\left(\hat{y}(\bm{x})\right)
\xrightarrow{M \to \infty} \sigma^2
\left(\frac{1}{n} + \frac{z^2}{\sum_{i=1}^n z_i^2}\right)  
\]

\noindent where $z = \bm{f}'\bs{\varphi}$ is the true  
projection for $\bm{x}$ and $\sigma^2$ is residual variance for 
 regression of $\bm{y}$ on $\bm{z}$.
\end{prop}

\begin{proof} 
Note $\bar z = \bm{0}$ and   $\mr{var}(\hat y(\bm{x})) =
\mr{var}(\hat\alpha) +   \bm{f}'\mr{var}(\bs{\hat \varphi}\hat\beta_{\bm{\hat
z}})\bm{f}$    where $\hat\beta_{\bm{\hat z}}$ is OLS slope on    $\bm{\hat z}
= \bm{F}\bs{\hat \varphi}$.   From Proposition \ref{mlevar} and 
the continuous
mapping theorem we have
$\bs{\hat\varphi}\overset{\mr{p}~}{\to}     
\bs{\varphi}$ and $\hat\beta_{\bm{\hat z}}
\leadsto \hat\beta_{\bm{z}}$.  Slutsky's lemma yields
$\bs{\hat \varphi}\hat\beta_{\bm{\hat z}}      \leadsto
\bs{\varphi}\hat\beta_{\bm{z}}$     with variance
$\bs{\varphi}\mr{var}(\hat\beta_{\bm{z}})\bs{\varphi}' =
\sigma^2\bs{\varphi}\bs{\varphi}'/\sum_i z_i^2.$   
Given that $\bs{\hat\varphi} \mapsto 
\bs{\hat \varphi}\hat\beta_{\bm{\hat z}}$ 
is bounded on its finite 
domain,  the Portmanteau lemma
implies our convergence.
\end{proof}

Thus, in our simple cartoon, MNIR-OLS approaches with {\it number-of-words}
the error rate of univariate least-squares.  This holds for  infill   (where
$n$ is constant but {speech-per-document} grows) as well as when $n$ is
growing with $M$ and the right-hand-side of \ref{olsvar} is decreasing.
Regularized estimation, say as applied in the main article, should help
efficiency in tougher setups (e.g., where vocabulary grows with $M$) but will
increase bias. Although we've focused on linear models many other
options are available -- for example, tree methods \citep[e.g.,][]{Brei2001}
work well in low dimensions for nonlinearity and variable
interaction.  The principles remain the same: results like Proposition
(\ref{mlevar}) show efficiency in collapsed IR, and one hopes to be able to
account for individual-level misspecification in the low dimensional forward
model.

\section{Latent factors}

Prof.~Blei's 2nd extension is an especially promising idea. Random effects
were originally viewed as a nuisance necessary  for  understanding
misspecification.  However,  a low-dimensional latent factorization of
these effects would be  a powerful tool for exploration and prediction.  It
provides a middle ground between  LDA and MNIR.

Such a model has log-odds $\bs{\eta} = \bs{\alpha} + \bs{\Phi}\bm{y} +
\bs{\Gamma}\bm{u}$  where $\bm{u} = [u_1\ldots u_K]'$ is a $K$-dimensional
factor vector.  $\bs{\Gamma}$ can then be interpreted as logit-transformed
LDA topics for variation in text not explained by variables in $\bm{y}$.
Just as $\bs{\Phi}'\bm{x}$ is sufficient for $\bm{y}$, the topic projection
$\bs{\Gamma}'\bm{x}$ will be sufficient for latent factors. Therefore the
model provides both a new way to think about latent structure in text
and a strategy for fast computation of  topic weights.

The difficulty with latent factor modeling is estimation. On the one hand,
although the model is more complex, estimation variance should still decrease
with $M$ because of the multinomial assumption on $\bm{x}$ (indeed, similar
arguments can explain the solid performance of LDA and sLDA
regression).  However, there are two big computational issues in posterior
maximization with document-specific  $\bs{\Gamma}\bm{u}_i$: you can no
longer collapse the likelihood, and you need to jointly solve for
$\bs{\Gamma}$ and $\bm{U}=[\bm{u}_1\ldots\bm{u}_n]'$.   Since the
discussants and I work on corpora many orders larger than the examples in this
article, additional latent structure is only useful if we can devise scalable
algorithms for its estimation.

On the lack of collapsibility,  which is also an issue for 
high-dimensional $\bm{y}$, I have had success  applying a MapReduce strategy
\citep{DeanGhem2004}.  A factorized likelihood is obtained by assuming counts
$x_{ij}$ and $x_{ik}$ for $j\neq k$ are independent and Poisson distributed
given $\bm{y}_i$ and $\bm{u}_i$  (centered on intensity $\exp(m_i/p)$ for
convenience).  The Map step groups counts on each column of $\bm{X}$ (i.e.,
for each word) and the Reduce step is a (possibly zero-inflated) Poisson log
regression of each word count onto $\bm{y}_i$ and $\bm{u}_i$. Exponential
family parametrization of the Poisson allows the same sufficiency results, and
the multinomial distribution for vectors of independent  Poissons given their
sum  implies a close connection to MNIR.  A paper on this  approach to
distributed multinomial regression is under preparation.

Even with these parallel algorithms,  it is difficult to solve for both
$\bm{U}$ and $\bs{\Gamma}$.  A  fixed-point solver (iterating between
maximization for each conditional on the other) is usually too slow.  One
could impute a rough guess for $\bm{U}$ (e.g., from a PCA of
document tf-idf), but this is only a stand-in solution.  Recent
advances in distributed optimization using ADMM \citep{Boyd2010} may offer a
way forward,  iterating from unique $\bm{U}_j$ for each $j^{th}$ word  towards
shared $\bm{U}$ across vocabulary, but this is just conjecture. The problem
of latent factor MNIR for large corpora remains unsolved. I look
forward to further discussion with Prof.~Blei on this because it is something
that his lab, if anybody, has a good chance of tackling.

\section{Interpretability}

Prof.~Grimmer's comments are focused on  interpretability:   the
translation from estimated models to scientific mechanisms. In particular,
he and other social scientists are interested in questions of {\it causation}.
This is among the  toughest of topics in statistics, and one that is only
growing in both difficulty and importance with the amount and dimension of
our data.

First, we should not underestimate the importance of predictive ability in
causal modeling.  The goal is always good prediction, but to understand
causation we want a model that predicts well when one covariate changes and
all others stay constant. Some of the best causal inference schemes are
explicitly predictive: matching, treatment-effects models, and propensity
scores rely upon estimation of the rate at which treated
individuals were assigned to that group.  As an example,  colleagues and I are
interested in measuring attribution for digital advertisements (i.e., how an
ad {\it causes} changes in consumer behavior).  This is a notoriously tough
problem, since the fact that a consumer sees an ad is highly correlated with
the likelihood that they were already looking to buy a certain product.  MNIR
for a consumer's text (e.g. on social media) and their browser history (where
website counts are treated like word counts) can be used to efficiently
predict the probabilities both that they see  an ad and that they buy a
product, and we hope to use this to disentangle these correlated outcomes.

However, instead of using text to help control for unobserved variables,
Prof.~Grimmer is seeking methods to infer the mechanisms behind word choice.
This is because he rightly wants to ensure that word loadings
correspond to a general notion of partisanship -- one that is portable between,
say, newspapers and congressional speech. This is the causal problem exploded
to simultaneous inference for thousands of correlated outputs. 
Regardless, MNIR is a natural starting point: I assume that
`sentiment' causes speech rather than the inverse. From this
one can look to apply the structural models used in econometrics and
biostatistics.  As mentioned, the effects of other inputs are `controlled for'
by including them in the log-odds, say as $\bs{\eta} =  \bs{\alpha} +
\bs{\varphi}y + \bs{\Theta}\bm{v}$ where $\bm{v} = [v_1 \ldots v_d]'$ are 
confounding variables. Going further, an MNIR treatment effects estimator
would  regress $y$ on $\bm{v}$ and  include the fitted expectation in the
equation for $\bs{\eta}$.   One needs to be careful here, as techniques used
for efficiency in high dimensions, such as sparse regularization, can bias
inference in unexpected ways.  See \citet{BellCherHans2012} for recent work on
sparse {high-dimensional} treatment effects estimation.

Finally, we should be aware of the limits of frameworks like MNIR (this also
relates to Prof. Blei's 3rd extension).  As Prof. Grimmer says, it {\it is}
difficult to know what covariates should be included or excluded from the
model.   However, this will always be as much of a problem in text analysis as
it has long been in social science. The `what' that we measure is only ever
defined in terms of observables and the model assumed around them (even with
human coders sentiment is dictated by the questions we ask).  The goal is
 to have this be as close as possible to our abstract ideal. For
example,  an ongoing project at Booth is investigating the history of
partisanship in congressional speech.  To define partisanship, we look at 
average predictability of party identity given words drawn from the 
distribution of speech for a given party.  The question of partisanship has
been transformed to one of predictability, and this notion is refined by
controlling for causes of word choice (e.g., geography, race) that we
understand as distinct from partisanship.  It is healthy to keep
this inference separate from abstract meanings for sentiment or partisanship,
in order to be clear on where evidence ends and speculation begins.

\vskip .5cm \noindent {\it Thanks to Jesse Shapiro, Matt Gentzkow, and Christian Hansen for helpful discussion.}

\setstretch{1}\small
\bibliographystyle{chicago}
\bibliography{taddy} 

\end{document}